\documentclass[letterpaper, 10 pt, journal, twoside]{IEEEtran}

\usepackage[utf8]{inputenc}
\usepackage{graphicx}  
\usepackage{setspace}  
\usepackage[titles]{tocloft}  
\usepackage{rotating}  
\usepackage[normalem]{ulem}  
 
\usepackage{amsmath}
\usepackage{amssymb}
\usepackage{bbm}
\usepackage[ruled,vlined]{algorithm2e}
\usepackage{color}
\ifCLASSOPTIONcompsoc
    \usepackage[caption=false, font=normalsize, labelfont=sf, textfont=sf]{subfig}
\else
\usepackage[caption=false, font=footnotesize]{subfig}
\fi

\usepackage{hyperref}
\hypersetup{
    colorlinks=false,
    linkcolor=blue,
    filecolor=magenta,      
    urlcolor=cyan,
    citecolor=black,
    pdftitle={Robust Model Predictive Path Integral Control},
    bookmarks=true,
    pdfpagemode=FullScreen,
}

\usepackage{float}
\usepackage[nolist]{acronym}
\usepackage[capitalise]{cleveref}
\usepackage{booktabs}  
\usepackage{numprint}  
\usepackage{wrapfig}
\usepackage{multicol}
\usepackage{multirow}

\newcommand{\rd}{{\mathrm d}}

\newcommand{\rT}{{\mathrm{T}}}

\newcommand{\E}{{\mathbb{E}}}

\newcommand{\vx}{{\bf x}}

\newcommand{\vp}{{\bf p}}

\newcommand{\vq}{{\bf q}}

\newcommand{\vu}{{\bf u}}

\newcommand{\vv}{{\bf  v}}

\newcommand{\vw}{{\bf w}}

\newcommand{\vF}{{\bf F}}

\newcommand{\valpha}{{\mbox{\boldmath$\alpha$}}}

\newcommand{\vtheta}{{\mbox{\boldmath$\theta$}}}

\newcommand{\argmin}{\operatornamewithlimits{argmin}}

\newcommand{\Qb}{\mathbb{Q}}

\newcommand{\Pb}{\mathbb{P}}
\newcommand{\Rb}{\mathbb{R}}

\newcommand{\KL}[2]{\mathbb{KL}\left({#1}\parallel {#2}\right)}

\newcommand{\ExP}[2]{\E_{{#1}}{\left[#2\right]}}

\newcommand{\pluseq}{\mathrel{+}=}

\usepackage{amsthm}

\newtheorem{lemma}{Lemma}

\usepackage{xcolor,colortbl}
\definecolor{Gray}{gray}{0.65}

\newcommand{\remove}[1]{\textcolor{red}{\sout{}}}

\title{Robust Model Predictive Path Integral Control: Analysis and Performance Guarantees}

\author{Manan S. Gandhi$^{1}$, Bogdan Vlahov$^{1}$, Jason Gibson$^{1}$, Grady Williams$^{1,2}$, Evangelos A. Theodorou$^{1}$%
\thanks{Manuscript received: October, 16, 2020; December, 28, 2020; Accepted January, 18, 2021.}
\thanks{This paper was recommended for publication by Editor Lucia Pallottino upon evaluation of the Associate Editor and Reviewers' comments.
This work was supported by Sandia National Laboratories, a multimission laboratory managed and operated by
National Technology and Engineering Solutions of Sandia, LLC., a wholly owned subsidiary of Honeywell International, Inc., for the U.S. Department of Energy’s National Nuclear Security Administration under contract DE-NA-0003525. Additionally, this work is supported by NASA LaRC.} 
\thanks{$^{1}$The Authors are with the Autonomous Control and Decision Systems Lab, Georgia Institute of Technology, Atlanta GA 30313 USA (\textit{Corresponding author: Manan Gandhi)} {\tt\footnotesize mgandhi3@gatech.edu; bvlahov3@gatech.edu; jgibson37@gatech.edu; evangelos.theodorou@gatech.edu})}%
\thanks{$^{2} $The Fourth Author is also with Embark Trucks, San Francisco CA 94107 {\tt\footnotesize Grady@embarktrucks.com} }%
\thanks{This paper has a supplementary video available to view at {\tt\footnotesize https://youtu.be/wQa3adL3AXU}}
\thanks{Digital Object Identifier (DOI): see top of this page.}
}

\begin{acronym}
\acro{MPPI}{Model-Predictive Path Integral Control}
\acro{Tube-MPPI}{Tube Model-Predictive Path Integral Control}
\acro{RMPPI}{Robust Model-Predictive Path Integral Control}
\acro{CCM}{Control Contraction Metrics}
\acro{MPC}{Model-Predictive Control}
\acro{IS}{Importance Sampling}
\acro{KL}{Kullback–Leibler}
\acro{GT-ARF}{Georgia Tech Autonomous Racing Facility}
\acro{iLQG}{iterative Linear Quadratic Gaussian}
\end{acronym}

\begin{document}
\maketitle


\begin{abstract}

In this paper we propose a novel decision making architecture for \ac{RMPPI} and investigate its performance guarantees and applicability to off-road navigation. Key building blocks of the proposed architecture are an augmented state space representation of the system consisting of nominal and actual dynamics, a placeholder for different types of tracking controllers,  a safety logic for nominal state propagation, and an importance sampling scheme that takes into account the capabilities of the underlying tracking control.  Using these ingredients, we derive a bound on the free energy growth of the dynamical system which is a function of task constraint satisfaction level, the performance of the underlying tracking controller, and the sampling error of the stochastic optimization used within \ac{RMPPI}.  To validate the bound on free energy growth, we perform experiments in simulation using two types of tracking controllers, namely the iterative Linear Quadratic Gaussian and Contraction-Metric based control. We further demonstrate  the applicability of \ac{RMPPI} in real hardware using the GT AutoRally vehicle. Our experiments demonstrate that  \ac{RMPPI} outperforms MPPI and Tube-MPPI by alleviating  issues of the aforementioned model predictive controllers related to either lack of robustness or excessive conservatism. \ac{RMPPI} provides the best of the two worlds in terms of agility and robustness to disturbances. 
\end{abstract}



\IEEEpeerreviewmaketitle

\section{Introduction}
\acresetall{}
\IEEEPARstart{R}{eal} world problems that require aggressive control solutions in uncertain environments require fast trajectory planning, control, and tracking. Some examples include off-road ground vehicle racing~\cite{williams2018best}, aerial acrobatics~\cite{abbeel2010autonomous}, and drone racing~\cite{foehn2020alphapilot}. Historically, sampling-based \ac{MPC}, such as \ac{MPPI} can be a solution to these types of tasks~\cite{williams2017information}. However, the performance of the controller is often hinged on the control sequence around which we are sampling, also known as the \ac{IS} trajectory. A severe disturbance can change the initial state of the system into a region where the importance sampler is no longer valid. \ac{Tube-MPPI}~\cite{williams2018robust} offers a potential solution by incorporating a nominal state solution that is immune from catastrophic disturbances, however there is no theoretical guarantee of performance and in practice, the nominal state can diverge. \ac{RMPPI} explores the concept of robustness from an information theoretic perspective through \ac{IS} and nominal state propagation.
\remove{We provide a review of the performance of this family of Model Predictive Path Integral control methods through both simulation and real experiments to demonstrate disturbance handling, failure causes, verification of the free energy bound. Finally we employ a novel scheme inside the \ac{IS} to draw connections between the tracking controller and free energy growth bounds.}
The contributions of our work are summarized below.

\noindent \textbf{1.)} We derive and empirically validate the bound on free energy growth for \ac{RMPPI}. This free energy bound measures the ``dynamic limit'' of the system, i.e. how close the system is operating with respect to task failure. We show this performance guarantee both in simulation and in reality on the AutoRally platform.

\noindent \textbf{2.)} We present a novel extension of \ac{IS} that utilizes \ac{CCM} to validate the exponential tracking guarantee. In simulation, we demonstrate how the performance of a nonlinear tracking controller impacts the free energy of the system. 

\noindent \textbf{3.)} We provide a systematic comparison of MPPI, Tube-MPPI, and Robust-MPPI, highlighting how each algorithm handles disturbances. The failure cases of each algorithm is explored in detail.

The rest of the paper is organized as follows: \cref{sec:relatedworks} reviews related work in Robust MPC, and \cref{sec:mathbackground} provides a background on information theoretic sampling based control, and motivates the need for \ac{RMPPI}. \cref{sec:rmppiderive} provides theoretical results related to the derivation of RMPPI, and proves the numerical bound on the free energy growth. \cref{sec:results} discusses the experimental results in both simulation and hardware.
\section{Related Works}
\label{sec:relatedworks}
Robust \ac{MPC} methods historically focus on linear MPC methods, \cite{richards2005robust, oravec2015alternative} and may not generalize well to nonlinear systems. Some extensions of MPC methods utilize an ancillary controller to guarantee robustness to external disturbances~\cite{mayne2011robust}, but determining the tube size is difficult. Other non-MPC methods can leverage alternative controllers and adaptation to maintain robustness~\cite{lakshmanan2020safe}, but these controllers are typically conservative. The trade-off between robustness and fast tracking relies on tuning the adaptation rate and filter frequency, but ultimately these controllers are not predictive, and thus may be ill suited for aggressive control problems. Sampling based motion planners, such as \cite{luders2013robust}, even for off-road driving tasks \cite{lee2016robust}, generate motion plans that ignore nonlinearities in the dynamics, and rely on a tracking controller to execute the motion plan. Other methods utilize a library of motion primitives \cite{sakcak2019sampling}, but again there is a tradeoff between dynamic feasibility, pre-computed maneuvers, and computational resources. Sampling based control schemes lean on \ac{IS}~\cite{zhang2014applications} to maintain tractability of the control problem, but are victim to both stochastic disturbances and errors in the dynamics model. The evolution of Path Integral Control has been towards increased performance when dealing with unmodeled disturbances and hazardous environments \cite{williams2018information, arruda2017uncertainty}. 
\section{Mathematical Background}
\label{sec:mathbackground}
Consider a general nonlinear system with discrete dynamics:
\begin{align}
    \vx_{t+1} &= \vx_t + \vF(\vx_t,\vu_t+\epsilon_t)\Delta t+\vw_t \label{eq:dyn} ,\\
        \mathcal{L}(\vx,\vu) &= q(\vx) + \lambda\vu^{\rT} \Sigma^{-1} \vu \label{eq:cost_function} ,\\
    S(U, \vx_0) &= \phi(\vx_T) + \sum_{t = 0}^{T-1} q(\vx_t) \label{eq:pathcost},
\end{align}
where $\vx\in\Rb^{n_x}$ is the state, $\vu\in\Rb^{n_u}$ is the control, $\epsilon \in\mathcal{N}(0,\Sigma)$ is the noise in the control channels, and $\vw\in\Rb^{n_x}$ is an external disturbance. State disturbances $\vw$ are unknown but assumed to be bounded. The function $\mathcal{L}(\vx, \vu)$ is the running cost function, with state cost $q(\vx)$ and quadratic control cost parameterized by inverse temperature $\lambda > 0$ and positive definite control cost penalty matrix $\Sigma^{-1}$. The function $S(U, \vx_0)$ is the state-to-path cost function, which takes in an initial condition $\vx_0$, set of inputs $U := \{\vu_0, ..., \vu_{T-1}\}$ and maps them to a cost value. $\phi(\vx)$ is the terminal cost. Next we define the free energy as follows:
\begin{equation}
\mathcal{F}(S, \Pb, \vx_0, \lambda) = -\lambda \log\ExP{\Pb}{\exp\left(-\frac{1}{\lambda}S(V, \vx_0) \right)} \label{eq:free_energy} ,
\end{equation}
where $V := \{\vu_0 + \epsilon_0, ..., \vu_{T-1}+\epsilon_{T-1}\}$ is the set of perturbed inputs. The expectation is taken with respect to $\Pb$, the probability distribution from which the controls of the system are sampled. In this case, the mean of the distribution is 0, representing the ``uncontrolled'' dynamics. Free energy is a metric for control performance in MPPI. This expectation can be estimated via \ac{IS} \cite{doucet2001introduction}. 

\noindent \textbf{Information-Theoretic MPPI:} \ac{MPPI} can be derived from the relation between the free energy and the relative entropy of two control systems. The two systems have their controls sampled from two distinct control distributions. One distribution is the ``controlled'' system $\Qb$ defined by density function: 
\begin{align}
    \vq(V) &= Z^{-1} \prod_{t=0}^{T-1}\exp \left(-\frac{1}{2} (\vv_t - \vu_t)^T \Sigma^{-1} (\vv_t - \vu_t) \right) \label{eq:density_func_controlled},
\end{align}
where $Z = \left( (2\pi)^m |\Sigma|\right)^{\frac{1}{2}}$. The other distribution is the baseline distribution $\mathbb{P}$ with density function $\vp(V)$ that typically represents the ``uncontrolled'' distribution with mean $\vu_t = 0$.
Through \ac{IS} and Jensen's inequality, we rewrite the free energy to be with respect to the controlled distribution:
\begin{align*}
\mathcal{F}(S, \Pb, \vx_0, \lambda) &\leq \ExP{\Qb}{S(V, \vx_0)} + \lambda \KL{\Qb}{\Pb}.
\end{align*}
We define the optimal distribution $\Qb^*$ that achieves the lower bound on cost in expectation with density defined as:
\begin{align*}
    \vq^*(V) = \frac{1}{\eta} \exp \left(-\frac{1}{\lambda} S(V, \vx_0) \right) \vp(V).
\end{align*}
It is impossible to sample from this optimal control distribution directly, but the KL-divergence between the two control distributions $\mathbb{Q}^*$ and $\mathbb{Q}$ can be utilized as an information theoretic objective to ``move'' the controlled distribution close to the optimal distribution \cite{williams2018information}. The problem becomes minimizing the KL divergence between the controlled distribution and the optimal controlled distribution: 
\begin{align}
    U^* = \argmin_U \KL{Q^*}{Q}.
    \label{eq:control_obj_MPPI}
\end{align}
In \cite{williams2018information}, the authors show that the solution to \cref{eq:control_obj_MPPI} is equivalent to:
\begin{align*}
    \vu_t^* = \int q^*(V) \vv_t dV.
\end{align*}
Using \ac{IS} to instead draw samples from our controlled distribution $\Qb$ provides the following:
\begin{align*}
    \vu_t^* = \ExP{\Qb}{w(V) \vv_t},
\end{align*}
where $w(V) = \frac{\vq^*(V)}{\vp(V)} \frac{\vp(V)}{\vq(V)}$ is the \ac{IS} weight. It can be shown that this weight is equivalent to:
\begin{align}
    w(V) = \frac{1}{\eta} \exp \left( -\frac{1}{\lambda} S(V, \vx_0) + \sum_{t=0}^{T-1} (\vv_t + \frac{1}{2} \vu_t)^T\Sigma^{-1}\vu_t\right) \nonumber.
\end{align}

MPPI has a key failure case that is addressed in \ac{Tube-MPPI}. The \ac{IS} trajectory is parameterized by an open loop control sequence, therefore a large disturbance can push the sampling into undesirable regions of the state space. The implicit assumption made in \ac{IS} is that the new state is close to the previous state, which breaks down in practical applications of the algorithm.

\noindent \textbf{Tube-MPPI:} 
An alternative setup is to keep the importance sampler focused around a nominal system that rejects disturbances. We consider two dynamical systems, one that represents the \textit{real} system from  \cref{eq:dyn} that experiences disturbances $\vw$, and another that represents the \textit{nominal} system $\vx^*$ in \cref{eq:nominaldyn} that is immune from disturbances. Let $\vu^*$ be the nominal control for:
\begin{align}
    \vx^*_{t+1} &= \vx^*_t + \vF(\vx^*_t, \vu^*_t + \epsilon_t)\Delta t \label{eq:nominaldyn}.
\end{align}

If state divergence is encountered, a feedback controller is applied to track the nominal system. The nominal state $\vx^*$ is reset to the real state $\vx$ if the difference between the free energy calculated at the nominal state and the real state is less than $\alpha$. In this case $\alpha > 0$ represents a free energy threshold. The importance sampler of the nominal state is always used when calculating rollouts from both the nominal and the real state, but it is reset to the real system control trajectory when the nominal state is reset.

There are two main shortcomings of \ac{Tube-MPPI}. The first is that the nominal state is chosen independently of the real state, therefore they can easily diverge. 
State divergence is the primary failure case of \ac{Tube-MPPI}~\cite{williams2018robust}. The second shortcoming is that the importance sampler is ignorant of the ancillary feedback controller. The \ac{IS} trajectories do not reflect actual system behavior, thus providing a biased estimate of the optimal control. This shortcoming limits the performance of Tube-MPPI; with small disturbances, it will perform just as \ac{MPPI} by electing to reset the nominal state. A large disturbance will result in the nominal state continuing independently and relying entirely on the tracking controller to bring the real system back to the nominal. 

\noindent \textbf{Robust-MPPI Architecture:}
\ac{RMPPI} extends the idea behind \ac{Tube-MPPI} to create a new controller framework with four pieces.

\subsubsection{Augmented State Space}
 Samples are drawn from two systems, the nominal system \cref{eq:nominaldyn} and the real system \cref{eq:augmented_real}:
\begin{align}
    \vx_{t+1} = \vx_t + \vF(\vx_t, \vu_t + k(\vx_t, \vx^*_t) +  \epsilon_t)\Delta t + \vw_t  \label{eq:augmented_real}, \\
    ||\vx_{t+1} - \vx^*_{t+1}|| \leq \gamma^t ||\vx_0 - \vx_0^*||,~ 0 < \gamma < 1 \label{eq:exptracking} ~~\forall t.
\end{align}
\cref{eq:augmented_real} incorporates feedback control into the forward propagated dynamics. To generate the bound on free energy growth, the feedback controller, $k(\vx, \vx^*)$, is assumed to drive the two systems together with a contraction rate, $\gamma$, seen in \cref{eq:exptracking}.  

\subsubsection{Importance Sampling}
For the augmented state space, samples are generated by simulating the combined system forward using \cref{eq:nominaldyn,eq:augmented_real}. A noise profile, $\epsilon$, is used when propagating the augmented state space. The importance sampler defines the control sequence for \emph{both} systems and is taken from the \emph{nominal system}. From the augmented state space, only the real system trajectories are utilized to ensure an unbiased estimate of the optimal control. The interaction between the nominal and real systems motivates the need for a mixed cost, $\tilde{S}(U, \vx_0,\vx^*_0)$, further discussed in \cref{lemma:CostFunction}. This lemma will prove the equivalence of this new cost to the original cost of the nominal state. Next, \cref{lemma:AISweights} derives the \ac{IS} weights for the augmented dynamical system. These weights are used for computing the unbiased estimate of the optimal control. The combination of this work results in \cref{Algorithm:robustMPPI/AIS}.
\begin{algorithm}[htbp!]
\footnotesize
\SetKwInOut{Input}{Given}
\Input{
      $\vF, k$: System dynamics and feedback controller\;
      $q$, $\phi$, $\Sigma$, $T$, $N$: Cost function and sampling parameters\;
      $\lambda, \alpha$: Temperature and cost threshold
      }
\SetKwInOut{Input}{Input}
\Input{
    $\vx_0, \vx_0^*$, $U$, $K$: Real/nominal state, IS sequence, feedback\;
}
\BlankLine
\For{$n \leftarrow 1$ \KwTo $N$}{
    $\vx \leftarrow \vx_0$;\quad $\vx^* \leftarrow \vx_0^*$;\quad $S_n, \hat{S}_n, S^{real}_n \leftarrow 0$\;
    Sample $\mathcal{E}^n = \left( \epsilon_0^n \dots \epsilon_{T-1}^n \right), ~\epsilon_t^n \in \mathcal{N}(0, \Sigma)$\;
    \For{$t \leftarrow 0$ \KwTo $T-1$}{
        $k_{fb} \leftarrow K_t(\vx, \vx^*)$\;
        $\vx \leftarrow \vF\left(\vx, \vu_t + \epsilon_t^n + k_{fb} \right)$\;
        $\vx^* \leftarrow \vF\left(\vx^*, \vu_t + \epsilon_t^n \right)$\;
        $\hat{S}_n \pluseq q(\vx) + \frac{\lambda (1-\beta)}{2} k_{fb}^\rT \Sigma^{-1} k_{fb}$\;
        $S_n \pluseq q(\vx^*)$ \;
        $S^{real}_n \pluseq q(\vx) + \frac{\lambda  (1-\beta)}{2}\left(\vu + k_{fb}\right)^\rT\Sigma^{-1}\left(\vu + 2\epsilon + k_{fb}\right)$
    }    
    $\hat{S}_n += \phi(\vx)$, $S_n += \phi(\vx^*)$, $S^{real}_n += \phi(\vx)$\;
    $S^{nom}_n = \frac{1}{2}S_n + \frac{1}{2}\max\left( \min \left(\hat{S}_n, \alpha \right), S_n \right)$\;
    \For{$t \leftarrow 0$ \KwTo $T-1$}{
        $S^{nom}_n \pluseq \frac{\lambda}{2}\sum_{t=0}^{T-1}\left( \vu^\rT \Sigma^{-1}\vu_t + 2\vu^\rT \Sigma^{-1}\epsilon_t^n\right) $\;
    }
}
\Return{$S^{nom}$,
        $S^\text{real}$,
        $\mathcal{E}$\;
}
\caption{Augmented Importance Sampler (AIS)} 
\label{Algorithm:robustMPPI/AIS}
\end{algorithm}

\subsubsection{Nominal System Propagation}
Unlike \ac{Tube-MPPI}, \ac{RMPPI} performs an optimization in order to find the best location for the nominal state. The goal of the nominal state is to remain as close to the actual state without resulting in a large increase in the free energy. Therefore we define a search around the current state and the previous nominal state (in a domain $\mathcal{D}$) such that the estimated free energy of the next nominal state is below our free energy threshold $\alpha$ in the following:
\begin{align}
\label{Equation:robustMPPI/NominalUpdatePractical1}
\vx^* &= \argmin_{\vp_i,~i \in \{0, 1, \dots R \}} \|\vp_i - \vx \|, \\
~~~~&s.t ~~~~ \mathcal{F}_{MC}(S, \Pb, \vp_i, \lambda) \le \alpha ~~~
\forall i \in \{0, 1, \dots R \}.
\end{align}
$\alpha$ is a hyperparameter that can have a physical meaning, usually the cost of crashing. $R$ is the number of candidates in $\mathcal{D}$. We approximate the free energy at each candidate by running an iteration of MPPI with a reduced number of samples from all $R$ locations. For computational reasons, we use a line search  between the previous nominal state ($\vx^*_{t-1}$), propagated nominal state ($\vx^*_{t}$) and, the real system state ($\vx_t$). This search is summarized by \cref{Algorithm:robustMPPI/NSP}.

As a result of the exponential tracking assumption between the real and nominal states and the solution of the nominal system propagation optimization problem, we are able to derive a bound on the growth of the free energy of the real system, shown in \cref{lemma:bound}. This bound can be used to determine how close the system is to task failure.
\begin{algorithm}[htbp!]
\footnotesize
\SetKwInOut{Input}{Given}
\Input{$\vF$, $\lambda$, $\alpha$: System dynamics, temperature, cost threshold\;
      $q$, $\phi$, $\Sigma$, $T$, $N$, $\beta$: Cost function and sampling parameters\;
      }
\SetKwInOut{Input}{Input}
\Input{
    $\vx_0, \vx_0^*$, $U$: Real/nominal state, IS sequence\;
}
\BlankLine
$\{\vp_0, \dots \vp_R \} \leftarrow Candidates(\vx_0, \vx^*_0, \vF, U)$\;
\For{$i \leftarrow 0$ \KwTo $R$}{
    $\vp^* \leftarrow \vp_i$\;
    \uIf{$i > 0$ }{ 
        \lFor{$t \leftarrow 1$ \KwTo $T-1$}{
            $\vu_{t-1}^i = \vu_t$
        }
        $\vu_{T-1}^i = 0$\;
    } 
    \lElse{
        $U^i = U$
    }
    \For{$n \leftarrow 1$ \KwTo $N$}{
        Sample $\mathcal{E}^n = \left( \epsilon_0^n \dots \epsilon_{T-1}^n \right), ~\epsilon_t^n \in \mathcal{N}(0, \Sigma)$\;
        \For{$t \leftarrow 0$ \KwTo $T-1$}{
            $\vp^* = \vp^* + \vF(\vp^*, \vu_t^i + \epsilon_t^n)\Delta t$\;
            $S_n \pluseq q(\vp^*) + \frac{\lambda(1-\beta)}{2}\left(\vu_t^\rT \Sigma^{-1}\vu_t + 2\vu_t^\rT \Sigma^{-1} \epsilon_t \right)$
            }
        $S_n += \phi(\vp^*)$
    }
    \lFor{$n \leftarrow 1$ \KwTo $N$}{
        $\eta \pluseq \exp\left(-\frac{1}{\lambda}S_n \right)$
    }
    $\mathcal{F}_i = -\lambda \log\left( \eta \right)$\;
}
$a = \argmin_i \left( \| \vp_i - \vx \| \right), ~~ s.t ~~ \mathcal{F}_i \le \alpha ~~$\;
\lIf{a = NULL}{a = 0 }
\Return{ $\vp_a$, $U^a$ \;}
\caption{Nominal State Propagation}
\label{Algorithm:robustMPPI/NSP}
\end{algorithm}

\subsubsection{Tracking Controllers}
The role of the tracking controller inside \ac{RMPPI} is to guide the importance sampler of the real system (which is subject to disturbances) to the nominal system (which is immune from disturbances). Ideally, this tracking controller will satisfy \cref{eq:exptracking}, however this property is difficult to achieve. Previously in \cite{williams2019model}, locally linear approximations of the dynamics were utilized along the nominal trajectory to then compute LQR gains for tracking by implementing \iac{iLQG} controller \cite{todorov2005generalized}.
However, this form of feedback does not provide the same guarantee for nonlinear systems outside of a small domain. This leads to the investigation of \acp{CCM} as an alternative feedback controller.

\acp{CCM} were introduced by \cite{manchester2017control} and are an example of a feedback controller that can guarantee exponential convergence for nonlinear systems, provided certain conditions are met. It is an extension of the contraction theory shown in \cite{lohmiller1998contraction}. The \ac{CCM} is designed for systems that are control and disturbance affine.

The two contracting trajectories in this work are the real and nominal trajectories:  $\{\vx_0, \vx_1, ...,\vx_{T-1}\}$ and $\{\vx^*_0, \vx^*_1, ...,\vx^*_{T-1}\}$, that are separated by $\delta_\vx$ 
which belongs to the tangent space $\mathcal{T}$ of the dynamical system at state $\vx$. In short, contraction theory designs a differential Lyapunov function candidate of the form: $V(\vx, \delta_\vx) = \delta_\vx^TM(\vx)\delta_\vx$, where $M(\vx)$ is a mapping from $\Rb^{n_x}$ to symmetric positive definite $n \times n$ matrices. For contracting systems with \emph{contraction rate} $\lambda$, the following expression is satisfied: $\dot{V}(\vx, \delta_x) \leq -2\lambda V(\vx, \delta_\vx),~ \forall \vx \in \Rb^{n_x},~ \forall \delta_\vx \in \mathcal{T}$. To reiterate, utilizing the CCM controller allows us to validate the assumption on exponential tracking, and tighten the bound on free energy growth, which will be shown in \cref{lemma:bound}.

\noindent \textbf{Real Time Performance:}
To ensure real time performance of the RMPPI controller, the sampling required for estimates of the free energy is done on a GPU. The same set of feedback gains is reused in the case of \ac{iLQG} for every sample to save computation; ideally each nominal trajectory would have its own optimal set of gains. Additionally, with the current form of the \ac{CCM} feedback controllers, it is not possible to implement the optimization within each thread of a GPU. As such, experiments for the CCM feedback are limited to CPU only and are \emph{not} real time.
\section{Theoretical Bound for Robust MPPI}
\label{sec:rmppiderive}
Next we prove three lemmas for the construction of \ac{RMPPI} and the bound on the real system free energy.

\noindent \textbf{Mixed Cost for the Nominal System:}
 The purpose of the mixed cost, \cref{Equation:robustMPPI/MaxMin}, is to penalize trajectories sampled from the nominal system which require large control effort from the feedback controller. Since the cost of the nominal system is subject to the upper bound $\alpha$, this mixed cost must also satisfy this bound.
\begin{lemma} Define the following augmented state-to-path cost function $\tilde{S}(V, \vx_0, \vx_0^*)$ as in the following:
\begin{align}
\label{Equation:robustMPPI/MaxMin}
\tilde{S}(V, \vx_0, \vx_0^*) &= \frac{1}{2}S(V, \vx_0^*)  \\ \nonumber
&~~+ \frac{1}{2}\max\left( \min \left(\hat{S}(V, \vx_0, \vx_0^*), \alpha \right), S(V, \vx_0^*) \right), \nonumber
\end{align}

where $\hat{S}$ is defined by
\begin{align}
\label{eq:s_hat}
\hat{S}(V, \vx_0, \vx_0^*) &= S(V, \vx_0) \\ \nonumber
&~~+ \frac{\lambda(1-\beta)}{2}\sum_{t=0}^{T-1} k(\vx_0, \vx_0^*)^\rT \Sigma^{-1} k(\vx_0, \vx_0^*).
\end{align}
Here $\lambda > 0$ is the inverse cost temperature, $\beta \in (0,1)$ is the control cost smoothing parameter, and $\alpha$ is the free energy threshold.
The augmented state-to-path cost function satisfies $\tilde{S}(V, \vx_0, \vx_0^*) \leq \alpha$ if and only if the state-to-path cost function for the nominal state satisfies $S(V, \vx_0^*) \leq \alpha$.
\label{lemma:CostFunction}
\end{lemma}

\begin{proof}
$S(V, \vx_0^*) < \alpha$ by construction. From \cref{eq:s_hat}, both terms on the RHS are positive, thus $\min \left(\hat{S}(V, \vx_0, \vx_0^*), \alpha \right) \leq \alpha$ for any $V, \vx_0$, and $\vx_0^*$. This implies:
\begin{align}
\label{eq:nominal_cost_proof_way_1}
    \tilde{S}(V, \vx_0, \vx_0^*) &\leq  \frac{1}{2}S(V, \vx_0^*)
+ \frac{1}{2}\max\left( \alpha, S(V, \vx_0^*) \right).
\end{align}
Using that $S(V, \vx_0^*) \leq \alpha$, we can see that \cref{eq:nominal_cost_proof_way_1} becomes $\tilde{S}(V, \vx_0, \vx_0^*) \leq \frac{1}{2}S(V, \vx_0^*) + \frac{1}{2}\alpha \leq \alpha$. This proves the statement is true in one direction.

Now suppose that $\tilde{S}(V, \vx_0, \vx_0^*) \leq \alpha$.
We have two cases to consider, with the first being when $\min \left(\hat{S}(V, \vx_0, \vx_0^*), \alpha \right) > S(V, \vx_0^*) $. That case gives us the following equation:
\begin{equation*}
    \tilde{S}(V, \vx_0, \vx_0^*) - \frac{1}{2} \min \left(\hat{S}(V, \vx_0, \vx_0^*), \alpha \right) = \frac{1}{2}S(V,\vx_0^*).
\end{equation*}
Since $\frac{1}{2}\min \left(\hat{S}(V, \vx_0, \vx_0^*), \alpha \right) > \frac{1}{2}S(V, \vx_0^*)$, we have:

\begin{equation*}
    \underbrace{\tilde{S}(V, \vx_0, \vx_0^*) - \frac{1}{2} \min \left(\hat{S}(V, \vx_0, \vx_0^*), \alpha \right)}_{=\frac{1}{2}S(V, \vx_0^*)} \leq \frac{1}{2} \tilde{S}(V, \vx_0, \vx_0^*),
\end{equation*}
and removing the $\frac{1}{2}$ from both sides gives $S(V,\vx_0^*) \leq \alpha$.

The other case is when $S(V,\vx_0^*) > \min \left(\hat{S}(V, \vx_0, \vx_0^*), \alpha \right)$. In that case, we have:

\begin{equation}
    \max\left( \min \left(\hat{S}(V, \vx_0, \vx_0^*), \alpha \right), S(V, \vx_0^*) \right) = S(V, \vx_0^*),
\end{equation}
which implies that $\frac{1}{2}S(V, \vx_0^*) + \frac{1}{2}S(V, \vx_0^*) =  S(V, \vx_0^*) = \tilde{S}(V, \vx_0, \vx_0^*) \leq \alpha$ which proves the result.
\end{proof}

\noindent \textbf{Importance Sampling Weights:}
Next we derive the \ac{IS} weights for the augmented dynamical system, \cref{eq:augmented_real}, which take into account the action of the feedback controller.
\begin{lemma}
\label{lemma:AISweights}
Let $\Qb_A$ denote the probability distribution defined by sampling from the augmented system dynamics, and let $\Pb$ be the distribution defined by the uncontrolled system. Then the \ac{IS} weight (Radon-Nikodym derivative) takes the form:
\begin{align*}
\frac{\rd \Pb}{\rd \Qb_A} &= \exp \bigg( -\frac{1}{2}\sum_{t=0}^{T-1} \left( \vu_t + k(\vx_t, \vx_t^*) \right)^\rT \Sigma^{-1} \nonumber \\ 
&\left( \vu_t + k(\vx_t, \vx_t^*) + 2\epsilon_t\right) \bigg).
\end{align*}
\end{lemma}
\begin{proof}
 
Probability distribution $\Qb_A$ is defined as: 
 \begin{align*}
    \vq_A(V) =& \left( (2\pi)^m |\Sigma|\right)^{-\frac{1}{2}} \prod_{t=0}^{T-1}\exp \bigg(-\frac{1}{2} (\vv_t - \vu_t + k(\vx_t, \vx_t^*) )^T \label{eq:density_func_augmented}\nonumber \\ 
    & \Sigma^{-1} (\vv_t - \vu_t + k(\vx_t, \vx_t^*)) \bigg) .
\end{align*}
 
\
 The ratio of the two distributions is equal to:
\begin{equation}
\underbrace{\frac{\rd \Qb_A}{\rd \Pb}}_{=\star} = \frac{
\exp\left( -\frac{1}{2}\sum_{t=0}^{T-1} \valpha_t^\rT \Sigma^{-1} \valpha_t \right)
}{\exp\left( -\frac{1}{2} \vv_t^\rT \Sigma^{-1} \vv_t \right)}, \nonumber
\end{equation}
where $\valpha_t = \vv_t - (\vu_t + k(\vx_t, \vx_t*)) $ 
Combining exponential terms, expanding multiplication, and cancelling like terms can simplify the equation:
\begin{align*}
    \star =& \exp \bigg( -\frac{1}{2} \sum_{t=0}^{T-1} \big(\vu_t + k(\vx_t, \vx^*_t)\big)^T\Sigma^{-1}\big(\vu_t + k(\vx_t, \vx^*_t)\big)) \\& - 2 \big(\vu_t + k(\vx_t, \vx^*_t)\big)^T\Sigma^{-1}\vv_t \bigg).
\end{align*}
Rewriting $\vv_t$ in terms of zero-mean noise $\epsilon_t$, gives us $\vv_t = \vu_t + k(\vx_t, \vx_t^*) + \epsilon_t$ and the following simplification:
\begin{align*}
    \star =& \exp\bigg(\frac{1}{2} \sum_{t=}^{T-1} \big(\vu_t + k(\vx_t, \vx^*_t)\big)^T\Sigma^{-1} \big(\vu_t + k(\vx_t, \vx^*_t) + 2\epsilon_t \big) \bigg).
\end{align*}
For un-biasing the expectation estimated with samples from $\Qb_A$, we need the inverse of this derivative $\frac{\rd \Pb}{\rd \Qb_A}$, which can be found by simply adding a negative to the exponential term. This completes the proof.
\end{proof}
\noindent \textbf{Bounding Free Energy Growth:}
The final lemma describes the bound on the \emph{growth} in free energy for the real system. This bound provides a measure of how close the system is to task failure.

\begin{lemma} Let $\mathcal{F}_{MC}(S, \Pb, \vx_0, \lambda)$ be the estimate of the free energy of the real system, $\mathcal{F}_{MC}(S, \Pb, \vx_0^*, \lambda)$ be the estimate of the free energy of the nominal system, $E_M^V$ be the upper bound on the uncertainty of the free energy estimates, $L_q$ be the Lipshitz constant of the state cost function $q(\vx)$, $L_{\phi}$ be the Lipshitz constant of the terminal cost function $\phi(\vx)$, $\gamma$ be the convergence rate of the controller, and denote $\| \vF(\vx_0, \vu) - \vx_0\| + \|\vx_0^* - \vx_0 \|$ as  $D_\vF(\vx_0, \vx_0^*, \vu)$. Then the growth of the free energy of the real system is bounded by:
\begin{align}
\Delta \mathcal{F}_{MC}(S, \Pb, &\vx_0, \lambda) ~~\le \left(\alpha -\mathcal{F}_{MC}\left(S, \Pb, \vx_0^*, \lambda\right)\right) + 2E_M^V \nonumber \\
&+ \left(L_{\phi} \gamma^T + L_q \frac{1 - \gamma^T}{1 - \gamma}\right)D_\vF(\vx_0, \vx_0^*, \vu) . \label{eq:bound_growth_rfe}
\end{align} 
\label{lemma:bound}
\end{lemma}
\begin{proof}
Assume we have the dynamics from \cref{eq:nominaldyn,eq:augmented_real,eq:exptracking}.
Via \ac{IS}, we can sample from \cref{eq:free_energy} with the controlled distribution $\Qb_A$:
\begin{align}
\label{eq:free_energy_with_q}
\mathcal{F}(S, \Pb, \vx_0, \lambda) = - \lambda\log\ExP{\Qb_A}{\exp\left(-\frac{1}{\lambda}S(V, \vx_0) \right) \frac{\rd \Pb}{\rd \Qb_A}}.
\end{align}

\cref{eq:pathcost} can be reformulated to utilize the nominal state as follows:
\begin{equation}
S(V, \vx_0) = \phi(\vx_T^* + e_T) + \sum_{t=0}^{T-1} q(\vx_t^* + e_t), \nonumber 
\end{equation}
where $e_t = \vx_t - \vx_t^*$. Using Lipshitz constants $L_\phi$ and $L_q$ we have:
\begin{align*}
S(V) &= \phi(\vx_T^* + e_T) + \sum_{t=0}^{T-1} q(\vx_t^* + e_t) \\
& \le \phi(\vx_T^*) + \sum_{t=0}^{T-1} q(\vx_t^*) + L_\phi \| e_T \| + \sum_{t=0}^{T-1}L_q \| e_t \|.
\end{align*}
For the rest of the analysis, we will use $ \vtheta = [L_{q},L_{\phi},\gamma]$. From the sum of a geometric series, we can bound the above to get:
\begin{equation*}
S(V, \vx_0) \le S(V, \vx_0^*) + \left(L_\phi \gamma^T + L_q \frac{1 - \gamma^T}{1 - \gamma}\right)\|\vx_0 - \vx_0^*\| .
\end{equation*}
Both the exponential and logarithm functions are monotonically increasing. Thus, utilizing \cref{eq:free_energy_with_q} allows us to upper bound the free energy of the real system:
\begin{align}
\mathcal{F} &\le -\lambda \log\ExP{\Qb_A}{\exp\left(-\frac{1}{\lambda}S(V, \vx_0^*) - \frac{1}{\lambda}   \delta_{1}(0; \vtheta) \right) \frac{\rd \Pb}{\rd \Qb_A}} , \nonumber
\end{align}
where the term $\delta_1$(t; \vtheta) is specified in:
\begin{equation*}
\delta_{1}(t; \vtheta)  =  \left(L_\phi \gamma^T + L_q \frac{1 - \gamma^T}{1 - \gamma}\right)\|\vx_t - \vx_t^*\|. 
\end{equation*}
Pull the constant terms outside of the expectation to give:
\begin{align*}
\mathcal{F} & \le  -\lambda \log \ExP{\Qb_A}{\exp\left(-\frac{1}{\lambda}S(V, \vx_0^*) \right) \frac{\rd \Pb}{\rd \Qb_A}}  + \delta_1(0; \vtheta) . 
\end{align*}
This implies the following relationship between the free-energy of the nominal and real system:
\begin{equation}
\mathcal{F}(S, \Pb, \vx_0, \lambda) \le \mathcal{F}(S, \Pb, \vx_0^*, \lambda)  + \delta_{1}(0; \vtheta) .\label{eq:real_nominal_FE}
\end{equation}
The real system free energy is bounded from above by the nominal system free energy plus a term depending on the slope of the cost function, state divergence, and the tracking controller convergence rate.\\
We can only compute estimates of the free energy, in this case, via Monte Carlo sampling. Let us define three terms: $e_M^B, \epsilon_M^{*V}, \epsilon_M^V\in\Rb$. $e_M^B$ represents the bias in the estimate, while $\epsilon_M^V$ and $\epsilon_M^{*V}$ are random variables sampled from the same distribution representing the variation in our estimators.
We will \emph{assume} that the bias $e_M^B$  is deterministic and constant, and that the variances $\epsilon^V_M$ and $\epsilon_M^{*V}$ are small and can be bounded with a deterministic positive constant $E^V_M\in\Rb$, i.e $|\epsilon_M^{*V}| < E^V_M$ and $|\epsilon_M^{V}| < E^V_M$. We can define the following relationship between the true free energy, and the Monte Carlo estimate of the free energy value for both the nominal and real systems respectively:
\begin{align}
    \mathcal{F}(S, \Pb, \vx_0^*, \lambda) = \mathcal{F}_{MC}(S, \Pb, \vx_0^*, \lambda) + e^B_M + \epsilon^{*V}_M \label{eq:nominal_true_mc_FE}, \\
    \mathcal{F}(S, \Pb, \vx_0, \lambda) = \mathcal{F}_{MC}(S, \Pb, \vx_0, \lambda) + e^B_M + \epsilon^V_M \label{eq:real_true_mc_FE} .
\end{align}
Since we compute the free energy of the nominal state and the real state with the \emph{same noise samples}, the bias of the free energy estimates will be the same; however the actual means will be different since these are two different states. Additionally, the variance of the two estimates will also differ.
If we have a Monte Carlo estimate below the free energy threshold for the current nominal state, $\mathcal{F}_{MC}(S, \Pb, \vx_0^*, \lambda) < \alpha$, we can show the following:
\begin{align} 
\mathcal{F}_{MC}(S, \Pb, \vx_0^*, \lambda) + e^B_M + \epsilon^V_M &= \mathcal{F}(S, \Pb, \vx_0^*, \lambda) ,\nonumber \\
\implies \mathcal{F}(S, \Pb, \vx_0^*, \lambda) &< \alpha + e^B_M + E^V_M \label{eq:bounded_nominal_FE}.
\end{align}
We see that above the true free energy is bounded above by the free energy threshold $\alpha$, the bias $e_M^B$ and the variance of the estimate $\epsilon^V_M$.  Substituting \cref{eq:real_true_mc_FE,eq:bounded_nominal_FE} into \cref{eq:real_nominal_FE} provides the following:
\begin{align}
    \mathcal{F}_{MC}(S, \Pb, \vx_0, \lambda) + e^B_M + \epsilon^V_M &< \alpha + e^B_M + E^V_M  + \delta_{1}(0; \vtheta), \nonumber \\
    \mathcal{F}_{MC}(S, \Pb, \vx_0, \lambda) < \alpha + 2E^V_M  &+ \delta_{1}(0; \vtheta) \label{eq:real_MC_FE_bound}.
\end{align}
%
This gives an upper bound on the estimate of the free-energy from the current real system state. 
Using $|\epsilon_M^{V}| < E^V_M$, we see that a $2E^V_M$ appears on the right hand side. Now, given the current estimate of the free-energies at the nominal and real states (denoted $\vx_0$ and $\vx_0^*$), we would like to provide a bound for what the free-energy estimate at the next real state will be (denoted as $\vx_1$). Using \cref{eq:real_MC_FE_bound}, we know that:
\begin{align}
\mathcal{F}_{MC}(S, \Pb, \vx_1, \lambda) < \alpha  + \delta_{1}(1; \vtheta) + 2E_M^V.
\end{align}
Inside $\delta_1(1, \vtheta)$, the worst case of propagation implies $\vx_1^*$ equals $\vx_0^*$:
\begin{equation}
\| \vx_1 - \vx_1^* \| \le \| \vx_1 - \vx_0^* \|, \nonumber
\end{equation}
The real system propagates with the disturbances, $\vx_1 = \vF(\vx_0, \vu) + \vw$,
where $\vw$ incorporates both the control dependent and any additional state-dependent noise. We assume that there is an upper bound for $\vw \le D$:
\begin{equation}
\| \vx_1 - \vx_0 \| \le \| \vF(\vx_0, \vu) \| + D .\nonumber
\end{equation}
We have that $\| \vx_1 - \vx_0^* \| = \| \vx_1 - \vx_0  - (\vx_0^* - \vx_0) \|$, and using the triangle inequality gives the following:
\begin{align*}
&\| \vx_1 - \vx_0  - (\vx_0^* - \vx_0) \|
\le \| \vx_1 - \vx_0 \| + \|\vx_1^* - \vx_0 \|\\
&= \| \vF(\vx_0, \vu) - \vx_0\| + \|\vx_0^* - \vx_0 \| + D.
\end{align*}
Denoting $\| \vF(\vx_0, \vu) - \vx_0\| + \|\vx_0^* - \vx_0 \|$ as  $D_\vF(\vx_0, \vx_0^*, \vu)$, we then have:
\begin{align}
\mathcal{F}_{MC}(S, \Pb, \vx_1, \lambda) &< \alpha  + 2E_M^V \nonumber \\
&+ \left(L_\phi \gamma^T + L_q \frac{1 - \gamma^T}{1 - \gamma}\right)D_\vF(\vx_0, \vx_0^*, \vu). \nonumber
\end{align}

Lastly, subtract $\mathcal{F}_{MC}(S, \Pb, \vx_0, \lambda)$ from both sides to complete the proof.
\end{proof}
The maximum increase in the estimate of the free energy is determined by three terms. Term one involves the user defined free energy threshold, which indicates the ``distance'' from constraint violation. Term two depends on the cost function slope, stiffness of the dynamics, and tracking controller convergence rate. The final term estimates the Monte Carlo sampling error.
The \ac{RMPPI} algorithm can be summarized in \cref{Algorithm:robustMPPI/R-MPPI}. 


\begin{algorithm}[h]
\footnotesize
\SetKwInOut{Input}{Given}
\Input{
    $\vF$, $N$, $\lambda$: System dynamics, samples, temperature\;
}
\SetKwInOut{Input}{Input}
\Input{
    $U_{init}$: initial control sequence
      }
\BlankLine
$U \leftarrow U_{init}$\;
$\vx^*, \vx \leftarrow StateEstimator()$\;
\While{alive}
{
$\{K_0, K_1, \dots K_{T-1} \} \leftarrow iLQR(\vx, \vx^*, U)$\;
        $\left\{S^{nom}, S^\text{real}, \mathcal{E}\right\}$
        $\leftarrow AIS(\vx, \vx^*, U, K)$\;
    $\vu_{opt} \leftarrow U_0 + K_0 (\vx - \vx^*) + \sum_{n=1}^N\frac{\exp\left(-\frac{1}{\lambda}S_n^{real}\right)\epsilon_0^n}{\sum_{n=1}^N\exp\left(-\frac{1}{\lambda}S_n^{real}\right)}$ 
    $U \leftarrow U + \sum_{n=1}^N \frac{\exp\left(-\frac{1}{\lambda}S_n^{nom}\right)\mathcal{E}^n}
    {\sum_{n=1}^N \exp\left(-\frac{1}{\lambda}S_n^{nom}\right)}$ 
    $SendControl(\vu_{opt})$\;
    $\vx \leftarrow StateEstimator()$\;
    $\vx^*,U \leftarrow NominalStatePropagation(\vx, \vx^*, U)$\;
}
\caption{\ac{RMPPI} with iLQG}
\label{Algorithm:robustMPPI/R-MPPI}
\end{algorithm}



\begin{table*}[hbtp!]
\footnotesize
\centering
\caption{Controller results in Autorally, each controller ran for approximately 15 laps.}
\begin{tabular}{lcccccc}
\toprule
{\textbf{Controller}}  & {\textbf{Mean Lap Time (s)}} & {\textbf{Mean Speed (m/s)}} & {\textbf{Fastest Lap}} & {\textbf{Top / Target Speed}} & {\textbf{Max slip (rad)}} & \textbf{Mean Lap Dist. (m)} \\
\midrule
MPPI Sim   & 30.16 $\pm$ 0.40 & 5.74 $\pm$ 0.92 & 29.50 & 6.83 / 7 & 0.24 & 173.35 $\pm$ 0.42  \\
Tube Sim   & 30.48 $\pm$ 0.37 & 5.65 $\pm$ 0.88 & 30.09 & 6.82 / 7 & 0.24 & 173.01 $\pm$ 0.60 \\
\textbf{RMPPI Sim}  & \textbf{29.05} $\pm$ \textbf{0.96} & \textbf{5.97} $\pm$ \textbf{1.49} & \textbf{27.88} & \textbf{8.84 / 9} & \textbf{0.65} & \textbf{173.51} $\pm$ \textbf{2.20} \\
\hline
MPPI Real  & 30.67 $\pm$ 0.53 & 5.77 $\pm$ 1.63 & 29.93 & 9.84 / 9 & 0.63 & 181.97 $\pm$ 1.13 \\
Tube Real  & 31.21 $\pm$ 0.43 & 5.68 $\pm$ 1.62 & 30.62 & 9.56 / 9 & 0.65 & 183.07 $\pm$ 0.94  \\
\textbf{RMPPI Real} & \textbf{31.07} $\pm$ \textbf{0.33} & \textbf{5.83} $\pm$ \textbf{1.88} & \textbf{30.74} & \textbf{10.31 / 9} & \textbf{0.84} & \textbf{186.28} $\pm$ \textbf{0.85}  \\
\bottomrule
\end{tabular}


\label{tab:simulation_results}
\end{table*}
\raggedbottom

\section{Results}
\label{sec:results}
\noindent \textbf{iLQG vs. CCM Tracking:}
To compare the tracking controllers, we ran RMPPI on both a linear double integrator system and a nonlinear system given in Section VI of \cite{manchester2017control}. 
\remove{For the linear problem, the true system noise variance that is $100$ times that of the nominal system. For the nonlinear problem the true system noise has a variance of $150$ times that of the nominal system.}
For the linear problem, the true noise in the control channel is $100$ times that of what RMPPI expects. For the nonlinear problem, the true noise in the control channel has a variance of $150$ times that of what RMPPI expects.
The goal is to test the robustness of each controller to large, unknown disturbances. 

\begin{figure}[htbp!]
    \centering
    \includegraphics[width=0.45\textwidth]{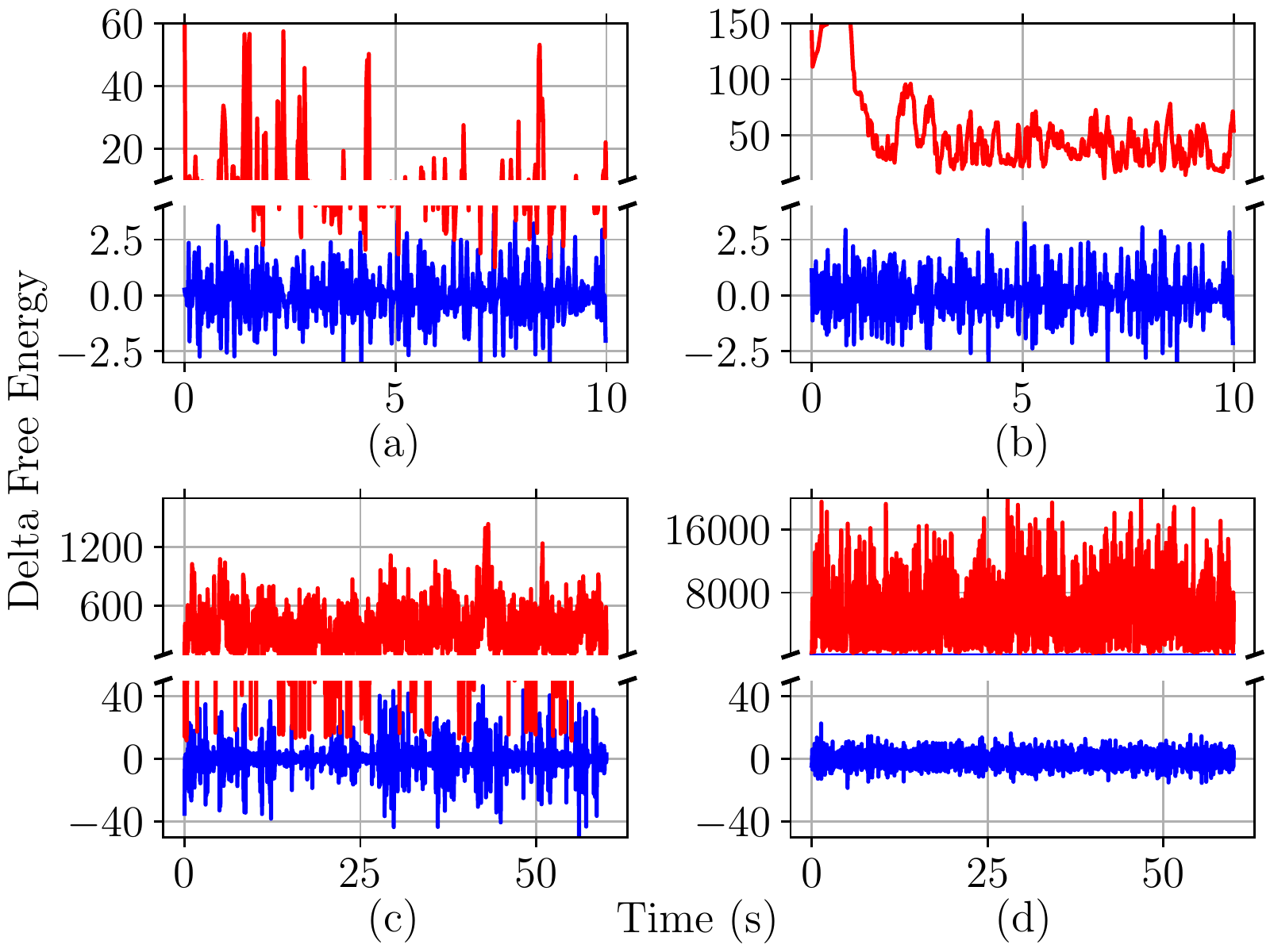}
    \caption{\textbf{Blue:} Free Energy Growth in \ac{RMPPI}. \textbf{Red:} Theoretical Bounds of Free Energy Growth for \ac{RMPPI}. (a)~uses \ac{CCM} as the tracking controller on the nonlinear dynamics, (b)~shows the nonlinear dynamics using \ac{iLQG}, (c)~shows the \ac{CCM} controller's performance on the double integrator system, and (d)~uses \ac{iLQG} on the double integrator. In both systems, the free energy growth bounds for the \ac{CCM} controller are tighter than that of the \ac{iLQG} controller due to the guaranteed convergence of the tracking controller in the task domain.}
    \label{fig:LQRvsCCM}
\end{figure}

We can see the free energy growth of the \ac{iLQG} system and the \ac{CCM} system in \cref{fig:LQRvsCCM}. Without a theoretical guarantee for exponential tracking, the worst case free energy growth bound from \ac{iLQG} is not very useful, as the magnitude of the upper bound is far higher than the actual free energy growth of the real system.  Given the fact that the CCM guarantees exponential tracking of the nominal state, we are able to compute \emph{useful free energy bounds}. These bounds provide a notion of how close the system is to the crash constraints. For aggressive maneuvering, it can be desirable to be able to recover from states and control trajectories that are near the free energy threshold. 
\remove{In terms of performance, we see a degradation in the ability to regulate free energy growth with the CCM controller because the controller is attempting to minimize control effort. This results in a slower rise time, and higher spikes in free energy. On the other hand, the \ac{iLQG} controller is tuned aggressively for a short rise time, and is able to handle the spikes in free energy while expending more control effort.} 
In both the linear and nonlinear dynamical systems, the magnitude of the free energy spikes is based on the aggressiveness of the underlying tracking controller. Controllers that are tuned similarly for disturbance rejection can show almost identical free energy handling properties, seen in \cref{fig:LQRvsCCM}a and \cref{fig:LQRvsCCM}b, with differing bounds. A controller tuned to minimize control effort may result in larger free energy spikes, as we see in the double integrator system for the CCM in \cref{fig:LQRvsCCM}c.

\noindent \textbf{AutoRally:} The AutoRally vehicle is an open source 1/5th scale autonomous racing platform \cite{AutoRallyPlatform}. We test our controller using both a physical platform and a simulated analog in Gazebo. All computation is run on board in real time. For the dynamics model of both the actual and simulated AutoRally platform, a feedforward Neural Network is used like in \cite{williams2017model}. 

We use a standard and robust cost function described in \cite{williams2019learning}.
The robust cost function has similar terms to the standard cost function but instead of a constant gradient it has constraint-like penalties with high slope. This formulation is piece-wise continuous and retains the Lipschitz continuity assumption on the cost function used in the theoretical bound. The main components of both AutoRally cost functions include a penalty for deviation from the center of the track, a slip penalty where the slip angle is $\arctan(\frac{v_x}{|v_y|})$, and a penalty for deviation from the target speed. We use \ac{iLQG} as the feedback controller on all iterations of MPPI on the AutoRally platform.


\subsubsection{AutoRally Simulation Experiments}
\begin{figure}[hbtp!]
    \centering
    \includegraphics[width=0.37\textwidth,clip,trim=20 15 75 35]{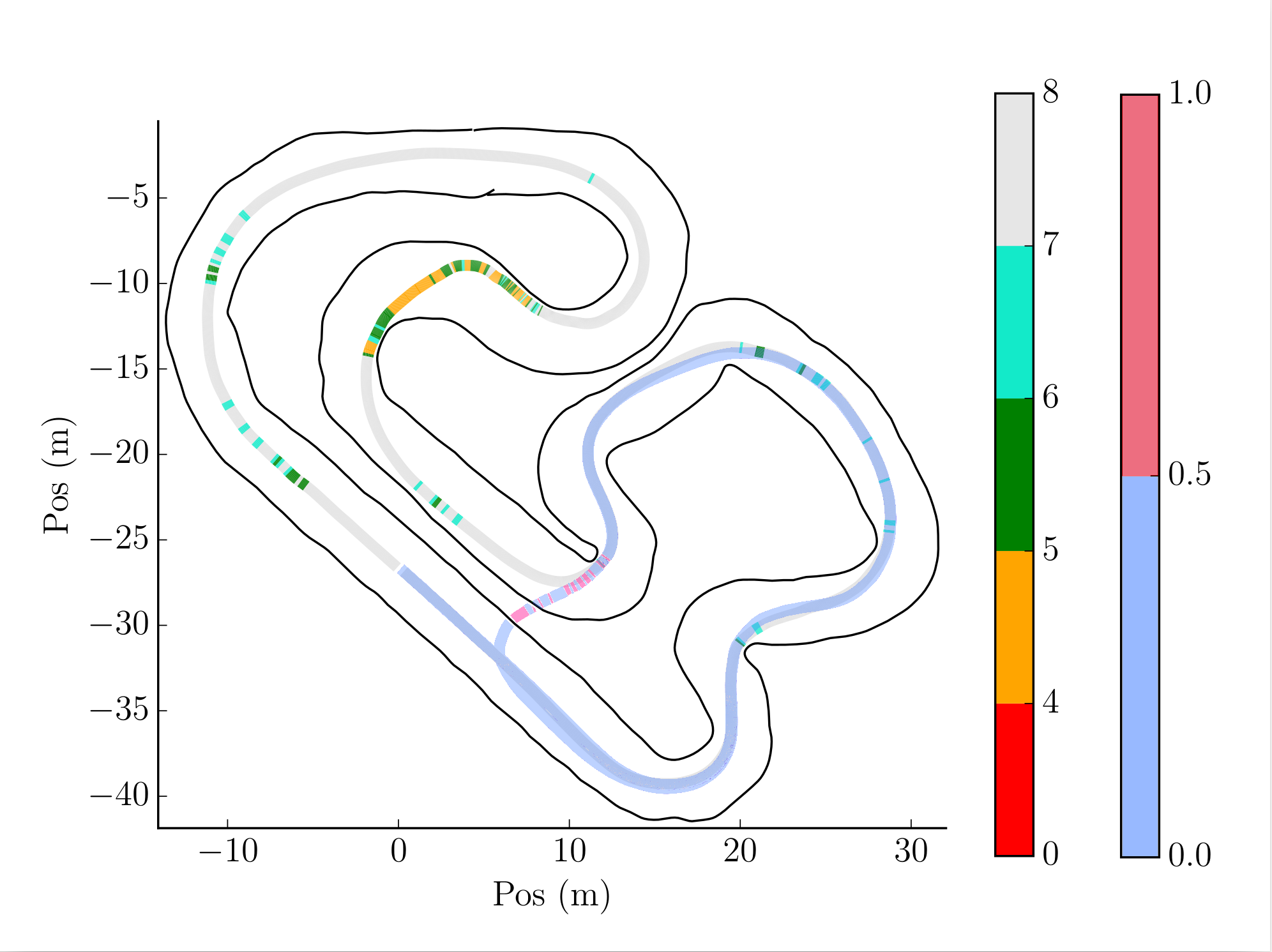}
    %
    \caption{Nominal State Selection of \ac{RMPPI} and \ac{Tube-MPPI} in the AutoRally simulator. For \ac{Tube-MPPI}, the nominal state is pink and the real state is light blue. For \ac{RMPPI}, State $0$ is $\vx^*_{t-1}$ or the previous nominal state, State $4$ is $\vx^*_t$ or the propagated nominal state, and State $8$ is $\vx_{t}$ or the real state. The other states are the line search interpolations between them.}
    \label{fig:MPPISim}
\end{figure}
The Neural Network dynamics model poorly represents the actual dynamics of the simulator. 
We ran each controller on the same simulated track with ground truth state and the same robust cost function. Each controller was set to the highest speed it could handle without consistently rolling over. The target speed of \ac{MPPI} and \ac{Tube-MPPI} were 7 m/s while \ac{RMPPI} was set to 9 m/s. The maximum slip angle for all controllers was set to $0.5$ rad. 

\cref{fig:MPPISim} shows the impact of the coupled nominal and real system in \ac{RMPPI} compared to \ac{Tube-MPPI}. \ac{Tube-MPPI} fails because as the real system begins to slide off of the track, the nominal system rejects the disturbance and continues to drive straight. At this set of pink points seen on the graph, the real system is sampling around a trajectory optimized for a different region of the state space and is entirely reliant on the external-to-MPPI application of \remove{the} feedback in order to correct back to the track. The optimization from the real state will not recover, while in the case of \ac{RMPPI}, the real system has feedback applied in the sampling and can sample around a trajectory that corrects for the disturbance. Since it is sampling the actual closed loop system, it intelligently corrects back to the track, rather than jumping the track boundary. This is seen with the yellow and green state selections in \cref{fig:MPPISim}. MPPI performs similarly to \ac{Tube-MPPI} but with slightly more erratic controls as it crosses the boundary. The reason that \ac{Tube-MPPI} is unable to reject the disturbance is due to the real system trajectory achieving a low enough cost by crossing the boundary and getting onto a different segment of the track which in turn causes the nominal state to be reset to that real state. The feedback controller does not steer the real state back to the nominal sufficiently fast enough to prevent this. 

An important note on the selection of an nominal state is that on average the real state is chosen, making \ac{RMPPI} and \ac{Tube-MPPI} behave similarly to MPPI. Only when there is a large disturbance is a different nominal state chosen and the additional benefits of \ac{Tube-MPPI} and \ac{RMPPI} realized. Notice that different nominal states are chosen along curves, when the vehicle is likely to encounter a disturbance.
\ac{RMPPI} never chooses a nominal state in the 0-3 range (line search from $x^*_0$ to $x^*_1$), which would indicate a loss of recursive feasibility. 

\subsubsection{AutoRally Hardware Experiments}

For state estimation, we are using a particle filter as described in \cite{Drews2019RAL}. 
\ac{Tube-MPPI} and \ac{MPPI} were both run with the standard cost function, and \ac{RMPPI} is using the robust cost function. The desired speed for all three controllers was set to 9 m/s, as this is the highest speed where \ac{MPPI} was able to reliably navigate the course. The maximum slip angle was set to 0.9 rad for all experiments.

\ac{RMPPI} has a lower cost for slip angle that results in more aggressive behavior around turns. However, even though \ac{RMPPI} had the capability to perform more aggressive maneuvers at lower costs than \ac{MPPI} and \ac{Tube-MPPI}, we do not see this result in faster lap times on the hardware. \ac{RMPPI} takes a longer racing line that results in a slower lap time even with a higher speed as shown in \cref{tab:simulation_results}. 
Overall, we see that \ac{RMPPI} achieves more aggressive behavior without violating the theoretical bound on a real system.

As seen in \cref{fig:MPPIRobustRealFreeEnergy}, the estimated free energy of the nominal state is typically less than that from the actual state, unless there is a large sampling bias \remove{\cref{fig:MPPIRobustRealFreeEnergy}}. The large spikes in free energy of the actual system when running \ac{RMPPI} represent possible failure cases of \ac{Tube-MPPI} and \ac{MPPI}. A disturbance could push the sampling of \ac{MPPI} into a region of high cost or could result in a state divergence that is not recoverable in the case of \ac{Tube-MPPI}. Both of these were observed during tuning of the controllers. Theoretically, unless \ac{Tube-MPPI} or \ac{RMPPI} chooses a different nominal state to use, we are effectively running \ac{MPPI}. 
This is validated in \cref{fig:MPPIVanillaTubeRealFreeEnergy} where most of the time, \ac{MPPI} has similar free energy dynamics to \ac{Tube-MPPI}. 
\begin{figure}[htbp!]
    \centering
    \subfloat[]{\includegraphics[width=0.25\textwidth, clip,trim=0 30 0 0]{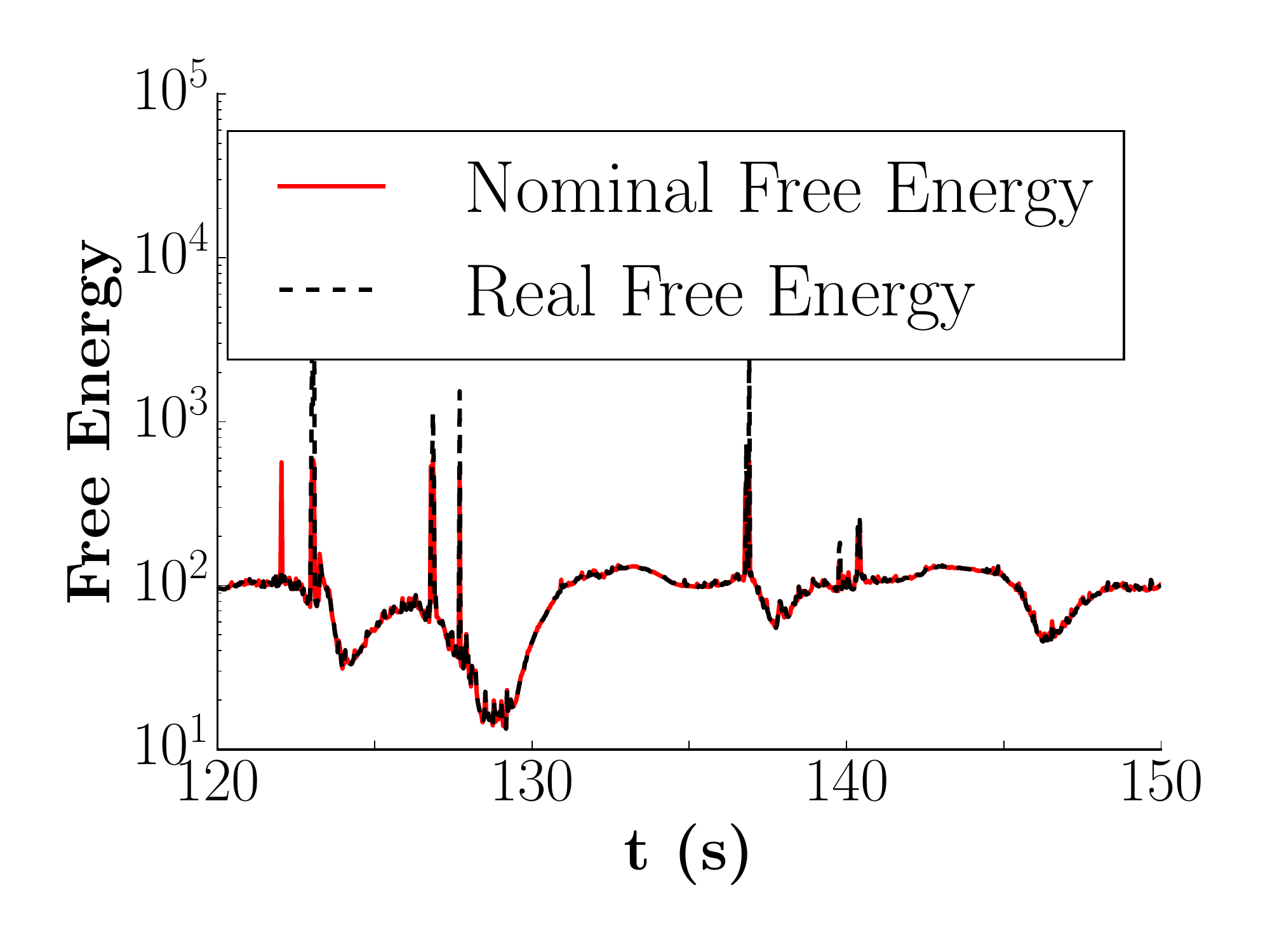}\label{fig:MPPIRobustRealFreeEnergy}}
    \subfloat[]{\includegraphics[width=0.25\textwidth, clip,trim=0 30 0 0]{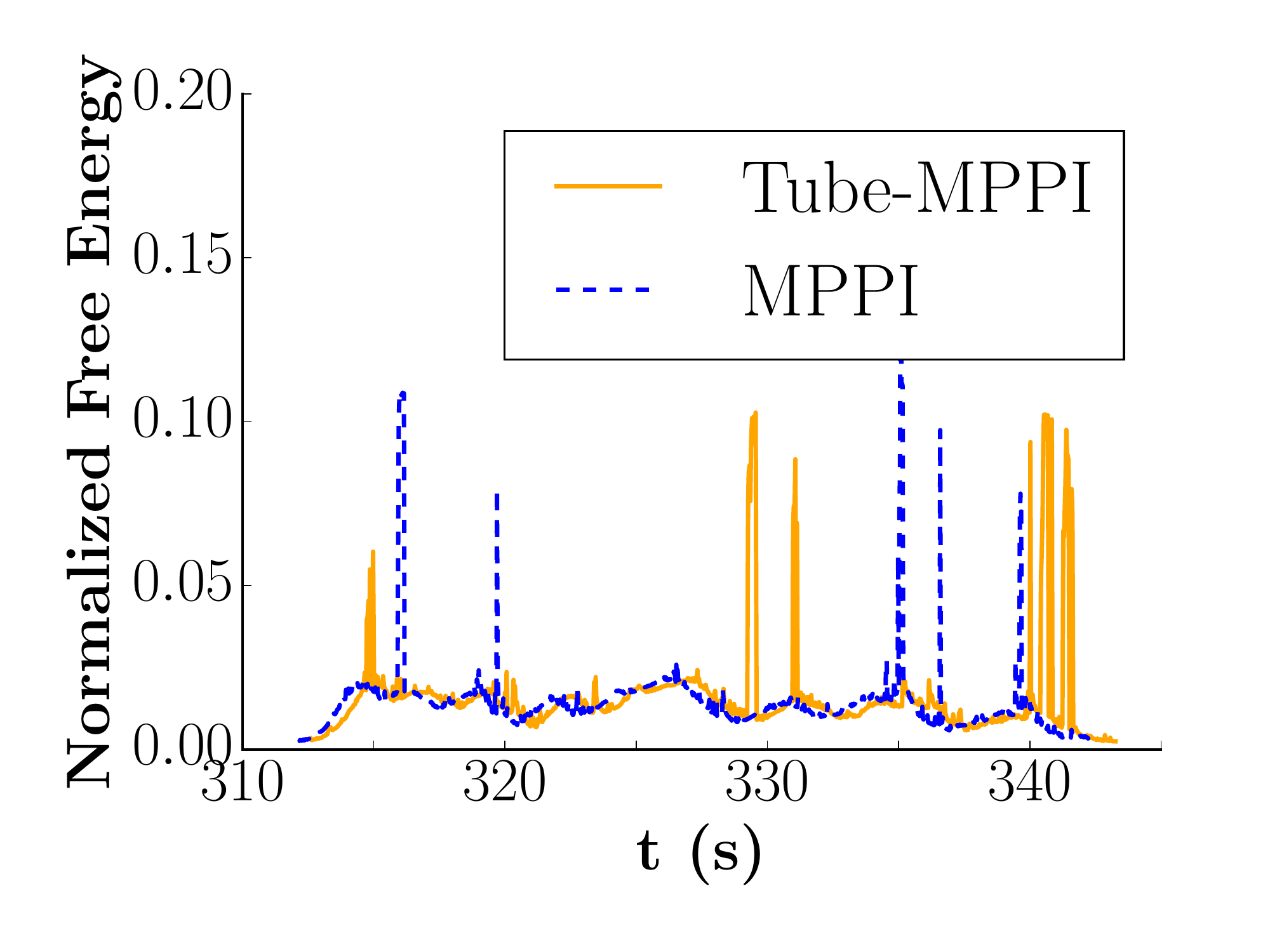}\label{fig:MPPIVanillaTubeRealFreeEnergy}}
    
    \caption{(a) Comparison of AutoRally Hardware Free Energy between RMPPI's real and nominal systems. The real free energy, on average, spikes higher than the nominal free energy since the nominal free energy is upper bounded by $\alpha$. (b) MPPI and Tube-MPPI real systems free energy spike similarly. In practice, we see that on the AutoRally platform, MPPI and Tube-MPPI perform similarly with the same cost function.}
    \label{fig:MPPIRealFreeEnergy}
\end{figure}

\section{Conclusions}
\label{sec:conclusion}


In this work, we propose a novel decision making architecture for \acl{RMPPI}. \ac{RMPPI} is characterized by its importance sampler and nominal state propagation that allows the system to gracefully recover from unknown disturbances, while simultaneously allowing the system to be controlled at its dynamic limits. We have validated the free energy bounds of RMPPI in both simulation and hardware, extended the underlying tracking controller to include recent developments in nonlinear contraction theory, and created a side-by-side comparison between RMPPI and the other forms of \acl{MPPI}. RMPPI has been shown to better control free energy and better handle large, unknown disturbances in the dynamics.






\bibliographystyle{IEEEtran}
\bibliography{references.bib}  

\end{document}